\documentclass[11pt,letter]{article}

\usepackage{fullpage}
\usepackage{color}
\usepackage{amsmath,amssymb,amsfonts,amsthm}

\newcommand{\prgph}[1]{\smallbreak\noindent{\bf\emph{#1}}}
\newcommand{\id}{\textsl{id}}
\newcommand{\idv}{\textsl{id}_{virt}}
\newcommand{\idvl}[1]{\textsl{id}_{loc}^{\,#1}}
\newcommand{\snap}{\texttt{snapshot}}
\newcommand{\local}{{\rm \textsf{LOCAL}}}
\newcommand{\slocal}{\textsf{SLOCAL}}
\newcommand{\congest}{\textsf{CONGEST}}
\newcommand{\dcpld}{{\rm \textsf{DECOUPLED}}}
\newcommand{\async}{{\rm \textsf{ABST-DECOUPLED}}}
\newcommand{\cent}{\textsf{CENTLOCAL}}
\renewcommand{\time}{\textsl{time}}
\newcommand{\dist}{\textsl{dist}}
\newcommand{\alg}{\textsf{\footnotesize ALG}}
\newcommand{\cA}{\mathcal{A}}
\newcommand{\cB}{\mathcal{B}}

\newcommand{\cF}{\mathcal{F}}
\DeclareMathOperator*{\argmin}{arg\,min}

\newtheorem{theorem}{Theorem}[section]
\newtheorem{lemma}{Lemma}[section]
\newtheorem{claim}{Claim}[section]

\newtheorem{definition}{Definition}[section]
\newcommand*{\affmark}[1][*]{\textsuperscript{#1}}

\begin{document}

\title{Distributed Computing in the Asynchronous LOCAL model\thanks{This research is supported by the ANR project DESCARTES (ref. DS0702-2016). Additional support from the INRIA project GANG. }}
\author{Carole Delporte-Gallet\affmark[1], Hugues Fauconnier\affmark[1], Pierre Fraigniaud\affmark[1], Mika\"el Rabie\affmark[2]\\~\\
\small \affmark[1]Institut de Recherche en Informatique Fondamentale \\
\small CNRS and Universit\'e de Paris, France   \\
\small {firstname.lastname@irif.fr}\\
\small \affmark[2]LIP6, Sorbonne Universit\'e, Paris, France \\
\small {mikael.rabie@lip6.fr}}
\date{}

\maketitle

\begin{abstract}
The \local\/ model is among the main models for studying locality in the framework of distributed network computing. This model is however subject to pertinent criticisms, including the facts that all nodes wake up simultaneously, perform in lock steps, and are failure-free. We show that relaxing these hypotheses to some extent does not hurt local computing. In particular, we show that, for any task $T$ associated to a locally checkable labeling (\textsc{lcl}), if $T$ is solvable in $t$ rounds by a deterministic algorithm in the \local\/ model, then $T$ remains solvable by a deterministic algorithm  in $O(t)$ rounds in an asynchronous variant of the \local\/ model whenever $t=O(\textrm{polylog}\, n)$. This improves the result by Casta\~neda et al. [TCS, 2019], which was restricted to 3-coloring the rings. More generally, the main contribution of this paper is to show that, perhaps surprisingly, asynchrony and failures in the computations do not restrict the power of the  \local\/ model, as long as the communications remain synchronous and failure-free. To this end, this paper introduces a new distributed renaming technique to provide nodes with consistent identifiers.

\medbreak

\end{abstract}


\section{Introduction}

\subsection{Locality in distributed network computing}
\label{subsec:locality}

Distributed network computing~\cite{P00} deals with the power and limitation of a collection of computing entities (a.k.a.~processes) occupying the nodes of a network, and exchanging messages along the links of this network.  In this framework, a primary interest has been placed on \emph{locality}, that is, determining what tasks can be solved whenever every process has to output after having exchanged information  with processes in its vicinity only, i.e., at bounded distance in the network. The \local\/ model~\cite{L92} has been extensively used for studying locality in network computing over the last 25~years~\cite{S13}. In this model, the network is modeled as a connected simple graph $G=(V,E)$, with processing nodes occupying the vertices of $G$, and communicating through the edges of $G$. Initially, every process is aware solely of its identity, which is supposed to be unique in the network. The \local\/ model is \emph{synchronous}: computation proceeds as a sequence of  \emph{rounds}, with all nodes starting at the same round. At each round, every node sends messages to its neighbors in~$G$, receives messages from its neighbors, and performs some individual computation. The round complexity of an algorithm is the number of rounds until all nodes output. For instance, a celebrated result in this context is Linial's lower bound~\cite{L92}  stating that 3-coloring the $n$-node ring requires at least $\frac12\log^*n-O(1)$ rounds. This bound is tight up to additive constants, thanks to Cole and Vishkin's algorithm~\cite{CV86}, which 3-colors the $n$-node ring in at most $\frac12\log^*n+O(1)$ rounds (see~\cite{RS15} for the exact constant additive factors). 

Moreover, the \local\/ model can be further simplified. Indeed, as pointed out in~\cite{L92}, a $t$-round algorithm $\cA$ in the \local\/ model can be simulated by another $t$-round algorithm $\cB$ which proceeds in two phases: first, each node collects all the data present at the nodes at distance at most $t$ around it, and, second,  each node individually simulates the behavior of the original algorithm $\cA$, without communication. In other words, a $t$-round algorithm in the \local\/ model can simply be viewed as a function from the ball $B_G(v,t)$ of radius $t$ around every node $v$ in $G$, to the output set. That is, a $t$-round algorithm in the \local\/ model can merely be viewed as: 
\begin{center}
\fbox{
\begin{minipage}{15cm}
\centerline{\sl Algorithm in the \local\/ model}
\begin{description}

\item[Phase 1:] Every node $v$ performs instruction $\snap(t)$, which returns to $v$ the structure of the ball $B_G(v,t)$ of radius $t$ centered at $v$ in $G$, with the identities of all the nodes in the ball. 
\item[Phase 2:] Node $v$ performs individual computations, taking solely as input the result of the snapshot instruction of Phase~1, and eventually outputs. 
\end{description}
\end{minipage}
}
\end{center}

This vision of the  \local\/ model  considerably simplifies the design of algorithms, and the analysis of the complexity of the problems. This probably explains why the  \local\/ model has been legitimately extensively used for a quarter of a century (see~\cite{S13} for a survey, and~\cite{BEG18,CLP18,CPS17,FHK16,GKM17,HSS18}  for a non-exhaustive selection of  more recent results). 

\subsection{Limits of the \local\/ model}

Despite all its positive aspects as far as studying local computing is concerned, the \local\/ model is subject to pertinent criticisms as far as practical applications are concerned. One such criticism is that the model assumes no bound on the computing power of the nodes, and on the throughput of the links. While this criticism is valid, it must be underlined that this apparent weakness of the model insures that lower bounds such as the one in~\cite{L92} are non conditional, i.e., they hold even if processes have infinite computing power, and even for full information protocols --- where every node is, at each round, forwarding all the knowledge that it accumulated during the previous rounds. Moreover, most of (but indeed not all) the upper bounds do not abuse this power, that is, most algorithms involve polynomial-time computation at the nodes, like the aforementioned Cole and Vishkin's algorithm~\cite{CV86}. Furthermore, the \congest\/ model~\cite{P00} has been designed especially for measuring the impact of limiting the bandwidth of the links for tasks involving high throughput, whose study in the context of the \local\/ model would be inappropriate. For instance, $C_4$ detection, i.e., determining whether the given network contains a cycle of length~4, is a trivial task in the \local\/ model, but it requires $\Theta(\sqrt{n})$ rounds to be solved in the \congest\/ model~\cite{DKO14}. 

Another criticism has been opposed to the \local\/ model: the fact that all nodes start at the same time, and proceed in lockstep. Indeed, in practice, the processes may proceed at different speed depending on various factors including heterogeneity of the CPUs, clock drifts, cache misses, poor load balancing, etc.  Moreover, the speeds of the different processes may vary with time, and processes may even be subject to all kinds of failures. In this context, the celebrated FLP theorem~\cite{FLP85} states that binary consensus cannot be solved in the asynchronous message passing systems, even if at most one process can crash. The argument opposed to the criticism about synchrony in the \local\/ model is the ability to use  \emph{synchronizers}~\cite{A85,AP90,PU89} for implementing synchronous algorithms in an asynchronous environment. However, while synchronizers are well suited to handle delays in the communications, they use \emph{waiting} mechanisms allowing each node to figure out when a round finishes. Such mechanisms are not suited for an environment in which processes can vary in speed and eventually crash. Indeed, waiting can cause deadlocks occurring when a process  waits forever for a process that has crashed. Instead, in asynchronous computing with an unbounded number of crashes, algorithms are required to be  \emph{wait-free}, that is, the algorithms must guarantee that every process can terminate and output correctly, independently from the behavior of the other processes. 

To sum up, one must admit that there is still a gap between the study of asynchronous crash-prone computing and the study of locality in network computing. The main issue addressed in the paper is therefore: \emph{Is making the \local\/ model slightly more realistic, e.g., by introducing some form of asynchrony, results in a weaker model, for which stronger lower bounds could be derived?} Surprisingly, the answer to this question is shown to be essentially negative: we show that introducing asynchrony in the computation {(while keeping synchronous communication between all nodes)} does not reduce much the power of the  \local\/ model, at least as far as the large class of tasks are concerned. 

\subsection{Decoupling computations from communications}

Casta\~neda et al.~\cite{CDFRR17}  has  initiated a line of work aiming at bridging the asynchrony-locality gap, by demonstrating that one can study locality in network computing even in the framework of asynchronous crash-prone processes. For this purpose, they introduced an asynchronous variant of the \local\/ model, called \dcpld, applied to symmetric networks (rings, toruses, etc.). This  latter model decouples the computing entities (processes) from the communicating entities (routers). The communications remain synchronous, that is, there is still a notion of rounds. However, the processes are fully asynchronous, and even subject to crash failures. In particular, the processes may wake up at different times. It is shown in~\cite{CDFRR17} that 3-coloring the $n$-node ring can still be done in $\frac12 \log^*n+O(1)$ rounds in the \dcpld\/ model.  I.e., for 3-coloring the $n$-node ring, it is sufficient that, upon wake up, every node $v$ gather messages produced during at most $\frac12 \log^*n+O(1)$ rounds, even if processes are fully asynchronous and subject to crash failures.

In this paper, we simplify the \emph{operational}  \dcpld\/ model from~\cite{CDFRR17}  into an \emph{abstract} model, called \async, that will be shown not stronger than the \dcpld\/ model (in symmetric networks), but easier to handle, and defined for all kinds of networks.  In a nutshell, one can view an algorithm in the  \async\/ model as performing in two phases, like in the standard \local\/ model (see Section~\ref{subsec:locality}). The first phase consists, for every awake process $v$, of taking a snapshot of the ball $B_G(v,t)$ of radius $t$ around $v$ in $G$. This snapshot returns the structure of the ball $B_G(v,t)$, and the identifiers of some processes in $B_G(v,t)$, depending on the wake up times of these processes. Roughly, for $v$ to get the identifier of a node $w\in B_G(v,t)$, it must be the case that $w$ woke up early enough so that a message from $w$ reaches $v$ no latter than $v$'s wakeup time, as this is the time at which $v$ takes its snapshot. The second phase consists of an individual computation performed at $v$, which eventually results in $v$ producing its output.  

The main difference between the \async\ model and the \local\/ model is the following. In the \async\/ model, if $w\in B_G(v,t)$ is not awake when $v$ is making its snapshot, then the identity of $w$ remains  unknown to~$v$. Instead, in the \local\/ model, a snapshot of $B_G(v,t)$ systematically returns the identifiers of all nodes in $B_G(v,t)$. 

Before presenting our results, we want to stress the fact that, in practice, the communication media (memory registers, wires, etc.)  can be viewed as synchronous. It is the \emph{access} to the these communication media that is asynchronous, due to various reasons, including scheduling and contention issues. The \dcpld\/ and \async\/ models are aiming at taking this phenomenon into account, as least some aspects of it, by keeping the communication synchronous, but assuming asynchronous computation. We also want to stress the fact that, as we shall discuss further in the text, although the \async\/ model includes a strong snapshot functionality, it is not stronger than the realistic operational \dcpld\/ model, at least as far as symmetric networks are concerned.  
 
\subsection{Our results}
\label{sec:ourresults}
 
We extend the results in~\cite{CDFRR17} to the entire class of \emph{locally checkable labeling} (LCL) tasks~\cite{NS95}, which is the main  center of interest of the research activities in the framework for local computing in networks. Many classical graph problems, e.g., vertex or edge-coloring, maximal matching, maximal independent set, minimal dominating set, etc., are LCL tasks. In a nutshell, an (input-free) LCL task is specified by a set $L$ of labels, and a family $\cF$ of balls with constant radius $r\geq 0$, in which every node is labeled by a label in $L$. For instance, $c$-coloring corresponds to the LCL task with $L=\{1,\dots,c\}$ and $\cF$ is the family of balls with radius~1, such that the label of the center is different from the labels of all its neighbors. 

In the  \async\/ model, solving an LCL task $(L,\cF)$ of radius~$r$ in a graph $G=(V,E)$ asks every correct process $v\in V$ to output a label in $L$ such that every ball $B_G(v,r)$ in which nodes corresponding to correct processes are labeled by their outputs, and nodes corresponding to processes that crashed are unlabeled, can be extended to a ball in $\cF$, by assigning labels in $L$ to the unlabeled nodes. 

We prove the following general result, which shows that, in the context of local computing, asynchronous crash-prone processes are essentially as efficient as reliable synchronous processes. A particular case of the statement below is when processes are initially aware of the size $n$ of the network, and that the identifiers are in $[0,n-1]$. 

\begin{lemma}\label{lem:main}
Let $(L,\cF)$ be an LCL task. Assume that,  in the \local\/ model, $(L,\cF)$ can be solved by a deterministic algorithm in $t(N)$ rounds in $n$-node graphs whenever the processes are initially aware of an upper bound  $N$ for~$n$, and that the identifiers are in the range $[0,N)$. Then,  in the \async\/ model,  $(L,\cF)$ can be solved by a deterministic algorithm  in at most $3\, t(N^2)$ rounds in $n$-node graphs, whenever the processes are initially aware of the upper bound  $N$ for $n$, and that the identifiers are in the range $[0,N)$.
\end{lemma}

In particular, every LCL task that can be solved in a polylogarithmic number of rounds in the \local\/ model, whenever  the processes are initially aware of an upper bound $N=O(poly(n))$ on the number of nodes, and  on the range of IDs,  can also be solved in a polylogarithmic number of rounds  in the \async\/ model. This applies to, e.g., maximal matching in arbitrary graphs~\cite{F17}, and $(\Delta+1)$-coloring graphs with constant max-degree~$\Delta$~\cite{FHK16}. In fact, this also holds for super-polylogarithmic round-complexities, such as the $2^{O(\sqrt{\log n})}$ upper bound for $(\Delta+1)$-coloring  $n$-node graphs with arbitrary max-degree~$\Delta$~\cite{PS92}. 

Since the \async\/ model will be shown not to be stronger than the \dcpld\/ model in symmetric graphs, we get the following, as a direct consequence of Lemma~\ref{lem:main}. 

\begin{theorem}\label{theo:main}
Let $(L,\cF)$ be an LCL task. Assume that,  in the \local\/ model, $(L,\cF)$ can be solved by a deterministic algorithm  in $t(N)$ rounds in symmetric $n$-node graphs whenever the processes are initially aware of an upper bound  $N$ for~$n$, and that the identifiers are in the range $[0,N)$. Then,  in the \dcpld\/ model,  $(L,\cF)$ can be solved by a deterministic algorithm in at most $3\, t(N^2)$ rounds in symmetric $n$-node graphs whenever the processes are initially aware of the upper bound  $N$ for~$n$, and that the identifiers are in the range $[0,N)$. 
\end{theorem}

In particular, a consequence of this theorem is that 3-coloring the $n$-node ring whose processes are given the initial knowledge of $n$ and that identifiers are in the range $[1,n]$, and where processes are sharing a consistent notion of clockwise and counterclockwise directions,  can be solved in $\frac32\log^*n+O(1)$ rounds in the \dcpld\/ model. The algorithm in~\cite{CDFRR17} performs in $\frac12\log^*n+O(1)$ rounds, but it is specific to 3-coloring the ring, while our theorem is generic, i.e., it applies to all LCL tasks. 

\subsection{Related work}

In addition to the aforementioned related work, it is worth mentioning several variants of the \local\/ model previously investigated in the literature~\cite{EMR18,GKM17,HKMS16,HKL+15}, among which the  \slocal\/ and the  \cent\/  models play an important role. 

In the \slocal\/ model~\cite{GKM17}, the nodes are processed in an arbitrary order. When a node $v$ is processed, it can see the current state of all the nodes in its $t$-hop neighborhood, for some $t\geq 0$. It can then compute its output as an arbitrary function of these states. The \slocal\/ model differs from the \async\/ model in many aspects. First, the \async\/ model considers arbitrary scheduling of the nodes, from  sequential scheduling (as in the \slocal\/ model), to parallel scheduling (as in the \local\/ model). Second, when a node is processed in the  \slocal\/ model, it has access to the identifiers of the nodes in its $t$-hop neighborhood, even if these nodes have not been processed yet. Instead, in the \async\/ model, a node has only access to the identifiers of the nodes  that are awaken in its $t$-hop neighborhood. In particular, the first nodes that wake up may have to output in complete ignorance of the identifiers of many (if not all) nodes in their $t$-hop neighborhood. This causes difficulties for transferring algorithms designed for the \local\/ model to the \async\/ model, whenever nodes heavily use the identifiers of the nodes in their vicinity for computing their outputs, like the Cole and Vishkin's algorithm~\cite{CV86}.  

In the \cent\/ model~\cite{EMR18}, the nodes are queried in arbitrary order, and a centralized (randomized) algorithm \alg\/ must answer each query by returning the status of each queried node in some (unknown) global solution. For instance, in the case of maximal independent set, \alg\/ must answer whether the queried node belongs to the independent set, or not. When a node is queried, \alg\/ can probe nodes in the vicinity of the queried node for forging its answer. It is  requested that the number of probed nodes should be sublinear in the size $n$ of the network. More importantly,  \alg\/ is also itself subject to severe restriction, and must run with a memory of sublinear size. Yet, \alg\/ must answer correctly to the queries, w.h.p., that is, all answers must be consistent w.r.t.~the task to be solved. The \cent\/ model obviously differs from the \async\/ model in many aspects, including the fact that, when a node is probed, it reveals its identity, whilst, in the \async\/ model,  only awake nodes reveal their identities. 

The ability of solving tasks such as consensus as a function of various parameters including processor synchrony, communication synchrony, message order, unicast vs. broadcast transmission, and atomicity of the send/receive operations has been the source of a vast literature (see, e.g.,~\cite{DDS87}). Our paper fits with this line of research. In the \dcpld\/ and \async\/ models, the processors are asynchronous, the underlying communication infrastructure is synchronous, the links are FIFO, and the transmissions  are unicast and atomic. 

Finally, it is also worth stressing that fault-tolerance has not been ignored in the context of network computing. In particular, self-stabilization deals with \emph{transient} failures, and is very much connected to local algorithms (see, e.g., \cite{BEG18}), and to local decision algorithms (see, e.g.,~\cite{BFPS14}). 

\section{The asynchronous  {\small LOCAL}  model(s)}

This section describes the \dcpld\/ model and  the \async\/ model. The former was introduced in~\cite{CDFRR17}, and the latter, introduced in this paper, is a variant of  the former, conceptually simpler to handle, and to analyse. We show that the \async\/ model is not stronger than the \dcpld\/ model in symmetric graphs. All the results further in this paper are then expressed in the  \async\/ model.

\subsection{The {\small DECOUPLED} model}
\label{sec:def-decoupled}

As for the \local\/ model, the \dcpld\/ model~\cite{CDFRR17} assumes $n$ nodes connected by a network modeled as a simple connected graph $G=(V,E)$, with nodes having distinct identities. We also assume that the edges of $G$ have port numbers (this does not add power to the \local\/ model). At every node~$v$, the  $\deg(v)$  edges incident to $v$ are labeled by  $\deg(v)$ distinct integers in the set $\{1,\dots,\deg(v)\}$, in a way similar to the model in~\cite{A80,YK96}.  The \dcpld\/ model differs from the \local\/ model in two aspects, by distinguishing the \emph{computation} layer from the \emph{communication} layer, and by assuming \emph{asynchronous computation}. 

A vertex $v\in V$ supports two entities: a \emph{router} $r_v$, and a \emph{process}~$p_v$. The former is in charge of routing messages passing through $v$, of delivering messages to the process $p_v$, and of forwarding messages emitted by $p_v$. Each process $p_v$ has an identifier, $\id(v)$, which is unique in the system. Routers do not have identifiers. However, at a degree-$d$ node~$v$, the $d$ communication links incident to the router $r_v$ are arbitrarily labeled with distinct integers from~1 to~$d$. These labels are called \emph{port numbers}. Note that an edge may have two different port numbers at its two extremities. These port numbers are solely for allowing every process to become aware of the structure of the network in its vicinity. In addition, port~0 designates the communication interface between the router~$r_v$ and the process~$p_v$. 

The \dcpld\/ model shares with the \local\/ model the fact that \emph{communications are synchronous}, that is, the communication network performs in \emph{rounds}. At each round, the messages placed in the input buffers of every router $r_v$ are forwarded through all  links incident to node $v$, i.e., placed in the input buffers of all routers $r_w$, for every node $w$ adjacent to $v$ in $G$. Any message placed in the input buffet of a router $r_w$ is immediately transferred to  the process $p_w$, during the same round. 

More specifically, the messages transferred  to a process $p_v$ are placed in an input buffer $b^{in}_v$, in which $r_v$ can write at every round, and in which $p_v$ can read at any point in time. Similarly, whenever $p_v$ aims at sending information through the network, it places these information in an output buffer $b^{out}_v$. At each round, the router $r_v$ takes all data stored in $b^{out}_v$ (if any), groups them as one message, and forwards this message to all neighbors of node~$v$. In particular, every message sent by a process at the beginning of a round is received by all neighboring processes at the end of the round, whenever the latter are awake. 

The processes $p_v$, $v\in V$, may perform at different speed, which may even vary along with time. In particular, each process wakes up at an arbitrary point in time, takes steps at arbitrary speed, and may even stop taking steps, in which case it is said to have \emph{crashed}. This behavior is often referred to as  \emph{wait-free} computing. However, while a process cannot wait for data from another process (since it cannot know whether the latter has crashed, or is just slow), every process has access to the clock governing the synchronous communications. Thus, in particular, a process can wait for any prescribed number of rounds. Once awake, every process can dequeue all messages available in its input buffer, if any, at any point in time, and fill up its output buffer at its own pace. 

Messages placed in the output buffer $b_v^{out}$ of every process $p_v$ are delivered by \emph{flooding} the network (every incoming message is sent through every outgoing link, but the one it arrived from).  Moreover, as for, e.g., the Internet, routers are capable of performing non-trivial operations on the messages, and, in particular, they can modify the message headers. In the  \dcpld\/ model, when a message is entering a router $r_v$ through port $i$, and is forwarded to port $j$ at the next round, the pair $(i,j)$ is systematically added to the header for keeping the history of the route followed by that message. Moreover, when a router $r_v$ dequeues $b_v^{out}$ for forming a message to be sent through the network, it tags this message with the round at which it is emitted. 

The \dcpld\/ model also assumes that the first instruction of a process $p_v$ is to place its identifier into its output buffer $b_v^{out}$. We denote by $\time(v)$ the round at which this event occurs. The message sent by $p_v$ is thus the pair $(\id(v),\time(v))$. Note that $p_v$ is aware of $\time(v)$. Then $p_v$ waits for a prescribed number of rounds, and accumulates messages of the form $(\id(w), \time(w))$ sent by nodes $w$ in the network. After this prescribed number of rounds, $p_v$ computes its outputs based on the received data. Observe that a process $p_v$ that wakes up late may receive messages $(\id(w), \time(w))$ from nodes $w$ that are at unbounded distance in the network. On the other hand, a process $p_v$ that wakes up early may receive no messages for an arbitrary long period of time, depending on when the other processes wake up, if they wake up at all (since they may have crashed).  

\begin{definition}[Casta\~neda et al. \cite{CDFRR17}]\label{def:complexity-decouppled}
The \emph{round-complexity} of an algorithm  in the \dcpld\/ model is the maximum, taken over all correct processes in the network, of the number of rounds the process waits before it is able to produce its output.   
\end{definition}

\subsection{The  \async\/ model}

We define the \async\/ model, which assumes the same setting as the \dcpld\/ model regarding asynchronous computing, but provides each process with an atomic \emph{snapshot} instruction that can be used only once, and that is the only communication mechanism. More specifically, as in the \dcpld\/ model, let $\time(v)$ be the round during which process $p_v$ wakes up. An algorithm in the \async\/ model is bounded to proceed in two phases identical to those performed in the \local\/ model, as described in Section~\ref{subsec:locality}.  

\begin{center}
\fbox{
\begin{minipage}{15cm}
\centerline{\sl Algorithm in the \async\/ model}
\begin{description}

\item[Phase 1:] Whenever a process $p_v$ wakes up, it performs instruction $\snap(t)$ at round $\time(v)$, for some $t\geq 0$, which may depend on~$n$ or on an upper bound $N$ of $n$. In a graph~$G$, this instruction instantaneously returns to $p_v$ the structure of the ball $B_G(v,t)$ of radius $t$ centered at $v$ in $G$, with the identifiers and wake up times of all processes $p_w$ satisfying 
\[
\dist_G(v,w)\leq t\;\;\;\mbox{and}\;\;\; \time(w)+\dist_G(v,w)\leq \time(v)+t.
\] 
\vspace{-5ex}
\item[Phase 2:] Process $p_v$ performs individual computations, taking solely as input the result of the snapshot instruction of Phase~1, and eventually outputs. 
\end{description}
\end{minipage}
}
\end{center}

\begin{definition}
The \emph{round-complexity} of an algorithm in the \async\/ model is the value of the parameter~$t$ involved in the snapshot operation.  
\end{definition}

The next lemma shows that the \async\/ model is not stronger than the \dcpld\/  model in symmetric networks. Recall that an \emph{automorphism} of a network $G=(V,E)$ is a one-to-one mapping $\phi:V\to V$ such that $\{x,y\}\in E \iff \{\phi(x),\phi(y)\}\in E$. An automorphism $\phi$ is said to be \emph{port-preserving} if, for every $x\in V$, and every $\{x,y\}\in E$, the port number at node $\phi(x)$ of the edge $\{\phi(x),\phi(y)\}$ is equal to the port number at $x$ of the edge $\{x,y\}$. 

\begin{definition}
A network $G=(V,E)$ is said to be  \emph{symmetric} if, for every two nodes $v,w \in V$, there exists a port-preserving automorphism $\phi$ of $G$ with $\phi(v)=w$. 
\end{definition}

Note that symmetric networks include most of the usual networks such as cliques, rings, toruses, hypercubes, and Cayley graphs in general. For instance, the networks considered  in~\cite{CDFRR17}, that is, rings in which every edge has distinct port numbers at its two extremities (i.e., there is a consistent notion of ``left'' and ``right''), are symmetric. A torus (i.e., a grid with wrapearound links) provided with a consistent notion of north, south, east, and west is also symmetric. A simple but crucial observation is that, if every node is given the initial knowledge that it is in a symmetric network (e.g., in a ring), it does not need to collect the \emph{structure} of the ball around it. On the other hand, every node may still need to collect the identifiers of the nodes in the ball around it. 

The following result show that the abstract \async\/ model is not stronger than the operational \dcpld\/ model, in symmetric networks. 

\begin{lemma}\label{lem:equivalence}
Assume that all nodes are initially aware that they belongs to a symmetric network~$G$. For any $t\geq 0$, any algorithm performing in $t$ rounds in $G$ under the  \async\/ model can be implemented to run in $t$ rounds in $G$ under the \dcpld\/ model. 
\end{lemma}

\begin{proof}
Let $\cA$ be a $t$-round algorithm in the  \async\/ model. This algorithm is transformed into an algorithm~$\cB$ in the \dcpld\/ model. In $\cB$, instead of performing the  operation $\snap(t)$ at wake up time, every process $p_v$ places the pair $(\id(v),  \time(v))$ in its output buffer~$b_v^{out}$, and waits for $t$ rounds. At round $\time(v)+t$, it collects all messages available in its input buffer $b_v^{in}$. After this point, $p_v$ executes the same instructions as those in Phase~2 of~$\cA$. 

To show correctness, it is sufficient to prove that the set of  messages available in the input buffer $b_v^{in}$ of $p_v$ after $t$ rounds contains all the information required by $p_v$ in  Phase~2 of~$\cA$. 

-- First, let $w$ be a node such that $\dist_G(v,w)\leq t$, and assume that $\time(w)\leq \time(v)+t-\dist_G(v,w)$. At round $\time(w)$, the pair $(\id(w),\time(w))$ started to be flooded from $w$. It reached $b_v^{in}$ at round $\time(w)+\dist_G(v,w)\leq \time(v)+t$. Therefore, the set of all messages available in the input buffer $b_v^{in}$ of $p_v$ after $t$ rounds of $\cB$ are those required by $p_v$ in Phase~2 of~$\cA$. 

-- Second, recall from Section~\ref{sec:def-decoupled} that, in the \dcpld\/ model, the port numbers of the communication links traversed by every message are appended to the message. As a consequence, since the network $G$ is symmetric, every node can compute where are the sources of the received messages positioned in~$G$. For instance, if $G$ is a torus, with links consistently labeled north (N), south (S), east (E), west (W) at every node, then a sequence of port numbers (e.g., SSWSW) corresponding to a message received at a node $v$ from a node $w$ allows $v$ to determine the position of $w$ in its ball of radius~$t$. 

It follows from the above that the set of all messages available in the input buffer $b_v^{in}$ of $p_v$ after $t$ rounds of $\cB$ contains all the information required by $p_v$ for performing Phase~2 of~$\cA$. That is, $v$ can reconstruct from these messages what would be the result of a snapshot of its ball of radius~$t$. This completes the proof of the lemma. 
\end{proof}

In the case of non-symmetric networks,  the \async\/ model may well be stronger than the \dcpld\/ model because it is not clear how to reconstruct the structure of a ball on the basis of the received messages in the \dcpld\/ model. Nevertheless, the reader should keep in mind that the \async\/ model is an abstract model that facilitates the analysis of the algorithms. It is by no mean claimed to be practical. Extending the \dcpld\/ model to arbitrary networks, beyond symmetric networks, is an interesting issue. 

\section{Proofs of Lemma~\ref{lem:main} and Theorem~\ref{theo:main} }

In this section, we prove the results stated in Section~\ref{sec:ourresults}, that is, Lemma~\ref{lem:main} and Theorem~\ref{theo:main}. 

\paragraph{Proof of Lemma~\ref{lem:main}.~}

The core of the proof of   is a reassignment of new identifiers to the nodes, so that every awake node can learn the new identifiers of all processes in its vicinity, including processes that are not yet woken up. These new identifiers are in the range $[0,N^2)$, whenever the original identifiers are in the range $[0,N)$. 

Let $(L,\cF)$ be an LCL task. Assume that there exists an algorithm $\cA$ solving $(L,\cF)$  in the \local\/ model, and assume that $\cA$ runs in $t(N)$ rounds in $n$-node graphs whenever the processes are initially aware of the upper bound  $N$ for $n$, and  that all identifiers are in the range $[0,N)$. Let $G=(V,E)$ be an $n$-node graph, and let $N$ be the upper bound on $n$ and on the range of identifiers given to the nodes. Let 
\[
\tau=t(N^2).
\]
First, we show how to assign \emph{local}  identifiers to the nodes in each ball of radius $\tau$. 

\begin{claim}\label{lem:IDlocales}
There exists a collection $\{\idvl{v}, v\in V\}$ of functions satisfying:
\begin{enumerate}
\item for every node $v\in V$, $\idvl{v}:B_G(v,\tau)\to[0,N^2)$ can be computed given solely $\id(v)$ and the structure of the ball $B_G(v,\tau)$; 
\item for every two nodes $v,v'\in V$, and for every two nodes $u\in B_G(v,\tau)$ and $u'\in B_G(v',\tau)$, if $\idvl{v}(u)=\idvl{v'}(u')$ then $v=v'$ and $u=u'$. 
\end{enumerate}
\end{claim}
\begin{proof}
Let us first address Item~1. Given the structure of the ball  $B_G(v,\tau)$, one can compute $\idvl{v}$ by enumerating the nodes in $B_G(v,\tau)$ according to some traversal of that ball.  For instance, one can perform a breadth-first search traversal of $B_G(v,\tau)$, traversing the incident edges with smaller port-numbers first. The nodes of $B_G(v,\tau)$ are then enumerated according to the order in which they are visited. Let $u_0,\dots,u_{k-1}$ be this enumeration, with $u_0=v$, and $k=|B_G(v,\tau)|$. We set 
\[
\idvl{v}(u_i) \; \overset{\mbox{\tiny def}}{=} \;  i + \id(v)\cdot N
\]
for every $i=0,\dots,k-1$. Since $0\leq \id(v) \leq N-1$ and $0\leq i < k\leq n \leq N$, Item~1 is satisfied.  

Let $v, v'\in V$, and let $u\in B_G(v,\tau)$ and $u'\in B_G(v',\tau)$. If $\idvl{v}(u)=\idvl{v'}(u')$, then, since $0\leq i < N$, we have $\id(v)=\id(v')$, and thus $v=v'$. Moreover, the index provided to $u$ by $v$  is equal to the index provided to $u'$ by $v$.  It follows that $u=u'$. Thus, Item~2 holds as well, which completes the proof of Claim~\ref{lem:IDlocales}. 
\end{proof}

Recall that, in the \async\/ model, if a process $p_v$ at node $v$ wakes up, and performs instruction $\snap(\varrho)$ for some integer $\varrho \geq 0$, at round $\time(v)$, it gets the structure of the ball $B_G(v,\varrho)$ of radius $\varrho$ centered at $v$ in $G$, with the identities and wake up times of all processes $p_w$ in this ball satisfying 
$
\time(w)+\dist_G(v,w) \leq \varrho+\time(v).
$
Let $\varrho$ and $\theta$ be two non-negative integers. For a fixed arbitrary scheduling of the processes, we define the \emph{snapshot sets} of a node $v$ as 
\[
S_v(\varrho,\theta) = \{\id(w): w\in V, \; \dist_G(v,w)\leq \varrho, \;\;\mbox{and}\;\; \time(w)+\dist_G(v,w) \leq \varrho+\theta\}.
\]
Note that if we would assume that $S_v(\varrho,\theta)$ also includes the structure of $B_G(v,\varrho)$, then we would have 
\[
\snap_v(\varrho)=S_v(\varrho,\time(v)).
\]
For the sake of simplifying notations, we sometimes  view $S_v(\varrho,\theta)$  as a set of nodes, rather than a set of node identifiers. Note that, for every node $v$, 
\[
\varrho'\geq \varrho \;\mbox{and} \; \theta'\geq \theta \; \Longrightarrow \; S_v(\varrho,\theta) \subseteq S_v(\varrho',\theta').
\]
Moreover, given $S_v(\varrho',\theta')$, one can extract $S_v(\varrho,\theta)$ out of it. We now show that there exists a function $\idv$ assigning (global) virtual identities to the nodes. 

\begin{claim}\label{lem:IDglobales}
There exists a function $\idv:V\to[0,N^2)$ satisfying: 
\begin{enumerate}
\item $\idv$ gives unique identifiers to the nodes, i.e., $\idv(v)\neq \idv(v')$ for every two nodes $v\neq v'$ in $G$, in the range $[0,N^2)$; 
\item for every $v \in V$, and for every $\theta\geq 0$, if $S_v(\tau,\theta)\neq\varnothing$, then $\idv(v)$ can be computed given solely the snapshot set $S_v(2\tau,\theta)$. 
\end{enumerate}
\end{claim}

\begin{proof}
Let us fix $v\in V$. For defining $\idv(v)$, we consider two cases. The first case is only for the purpose of completeness, as it does not correspond to operationally practical scenarios. It assumes that $S_v(\tau,\theta)=\varnothing$ for every integer $\theta\geq 0$. Under this assumption, we set 
\[
\idv(v) \; \overset{\mbox{\tiny def}}{=} \;  \idvl{v}(v) = \id(v)\cdot N.
\]
In the operationally relevant case where  there exists an integer $\theta\geq 0$ such that $S_v(\tau,\theta)\neq\varnothing$, we set 
\[
\theta^*\; \overset{\mbox{\tiny def}}{=} \; \min\{\theta \in \mathbb{N}: S_v(\tau,\theta)\neq \varnothing\}.
\]
and
\[
v^*\; \overset{\mbox{\tiny def}}{=} \; \argmin\{(\time(u)+\dist_G(u,v),\id(u)), u\in S_v(\tau,\theta^*)\}.
\]
That is, $v^*$ is the node in  $S_v(\tau,\theta^*)$ with smallest sum of wake up time plus distance to~$v$, where ties are broken using the identities. Then we finally set 
\[
\idv(v) \; \overset{\mbox{\tiny def}}{=} \;  \idvl{v^*}(v).
\]
That is, the virtual identifier of $v$ is the local identifier given to $v$ by the node $v^*$. 

\medskip 

We have $\idv(v)\neq \idv(v')$ for every $v\neq v'$ because, by Claim~\ref{lem:IDlocales}, for every four nodes $v,v',w,w'$,  the equality $\idvl{w}(v)= \idvl{w'}(v')$ implies $v=v'$. Thus, Item~1 holds. 

\medskip 

Let $v \in V$, and $\theta\geq 0$. Let us assume that $S_v(\tau,\theta)\neq\varnothing$, and let us be given the snapshot set $S_v(2\tau,\theta)$. Let
\[
v^*_\theta \; \overset{\mbox{\tiny def}}{=} \;  \argmin\{(\time(u)+\dist_G(u,v),\id(u)), u\in S_v(\tau,\theta)\}.
\]
We claim that 
\[ 
v^* =v^*_\theta.
\]
Indeed, we have $\theta^*\leq \theta$, and thus $S_v(\tau,\theta^*)\subseteq S_v(\tau,\theta)$. So, let $w\in S_v(\tau,\theta) \setminus S_v(\tau,\theta^*)$. Since $v^*\in S_v(\tau,\theta^*)$ and $w\notin S_v(\tau,\theta^*)$, we have 
\[
\time(v^*)+\dist_G(v,v^*) \leq \tau + \theta^* < \time(w)+\dist_G(v,w),  
\]
from which it follows that $v_\theta^*\neq w$. Therefore, $v_\theta^*\in S_v(\tau,\theta^*)$, and thus $v^* =v^*_\theta$, as claimed.

\medbreak

Given the snapshot set $S_v(2\tau,\theta)$, one can extract $S_v(\tau,\theta)$ out of it, and compute $v^* =v^*_\theta$. The ball $B_G(v^*_\theta,\tau)$ can be extracted from $S_v(2\tau,\theta)$ as well since $v^*_\theta\in B_G(v,\tau)$. By Claim~\ref{lem:IDlocales}, the identity of $v^*_\theta$, and the ball $B_G(v^*_\theta,\tau)$ are sufficient to compute $\idvl{v^*}(v)$, that is, to compute $\idv(v)$. Thus, Item~2 holds, which completes the proof of Claim~\ref{lem:IDglobales}. 
\end{proof}

We now have all ingredients to implement Algorithm~$\cA$ in the \async\/ model. At every node~$v$, the algorithm proceeds as follows: 
\begin{description}
\item[Phase 1:] When process $p_v$ wakes up, it performs $\snap(3\,\tau)$ at round $\time(v)$;
\item[Phase 2:] Then, process $p_v$ internally simulates a run of $\cA$ on $B_G(v,\tau)$ with the identifiers in $[0,N^2)$ provided by $\idv$ to the nodes in $B_G(v,\tau)$.  
\end{description}
Note that $p_v$ can  execute Phase~2. Indeed, $p_v$ can compute $\idv(v')$ for every node $v'\in B_G(v,\tau)$. This is because, by Claim~\ref{lem:IDglobales}, since $v\in S_{v'}(\tau,\time(v))$, it is sufficient for $v$ to know $S_{v'}(2\tau,\time(v))$, which is satisfied as  
\[
\snap_v(3\,\tau)\supseteq S_{v'}(2\tau,\time(v))
\] 
for every $v'\in B_G(v,\tau)$. 

It remains to establish the correctness of the above implementation of $\cA$ in the \async\/ model. This is immediate once we observe that the output of every node $v$ is equal to the output of $\cA$ at $v$ in the \local\/ model, whenever the identities of the nodes are not provided by the one-to-one function $\id:V\to [0,N)$ but by the one-to-one function $\idv:V\to [0,N^2)$. It follows that any partial solution provided by the set of correct processes can be extended to a global solution, by providing every faulty process with the label that this process would have produced if it had not crashed.  This completes the proof of Lemma~\ref{lem:main}. \qed

\bigskip

Theorem~\ref{theo:main} directly follows from Lemma~\ref{lem:main} and Lemma~\ref{lem:equivalence}. 

\section{Conclusion}

This paper is carrying on the line of research aiming at bridging the gap between the study of asynchrony in distributed computing, and the study of locality in network computing. The \async\/ model introduced in this paper simplifies the \dcpld\/ model defined in~\cite{CDFRR17}, and extends it to arbitrary networks. Using the  \async\/ model, we prove that, for LCL tasks (which include the vast majority of the problems  studied in local network computing),  asynchrony and crash failures do not hurt the power of computation of the \local\/ model, as long as asynchrony and failures impact computation only, while communications remain synchronous and failure-free. Incorporating asynchrony in the communications too can be done thanks to synchronizers. However, incorporating asynchrony \emph{and} crash failures in both computation \emph{and} communications appears to be challenging in the context of network computing. 

\bigbreak

\prgph{Acknowledgement:} \; The fourth author is thankful to Alkida Balliu for fruitful discussions in connection with the topic of this paper. 

\bibliographystyle{plain}


%
%
%

\end{document}